\documentclass[a4paper,UKenglish,cleveref, autoref, thm-restate]{lipics-v2021}

\pdfoutput=1 
\hideLIPIcs  

\usepackage{todonotes}
\usepackage{xspace}

\usepackage{paralist}
\newcommand{\decprob}[3]{
	\begin{center}
		\begin{minipage}{0.96\linewidth}
			\noindent
			\textsc{#1}
   \smallskip
			\begin{compactdesc}
			 \item[\textbf{Input:}]  #2
			 \item[\textbf{Question:}]  #3
			\end{compactdesc}
		\end{minipage}
	\end{center}
}

\newcommand{\optprob}[3]{
	\begin{center}
		\begin{minipage}{0.96\linewidth}
			\noindent
			\textsc{#1}
   \smallskip
			\begin{compactdesc}
			 \item[\textbf{Input:}]  #2
			 \item[\textbf{Task:}]  #3
			\end{compactdesc}
		\end{minipage}
	\end{center}
}

\newcommand{\probweighted}{$P\mid r_j, p_j = p \mid\sum w_j U_j$\xspace}
\newcommand{\probunweighted}{$P\mid r_j, p_j = p \mid\sum U_j$\xspace}

\usepackage{etoolbox}
\newcommand{\appsymb}{$\star$}
\newcommand{\appref}[1]{{\hyperref[proof:#1]{\appsymb}}}
\newcommand{\appLink}[1]{{\hyperref[#1]{\appsymb}}}

\newcommand{\appendixsection}[1]{%
\gappto{\appendixProofs}{
\section{Additional Material for~\cref{#1}}\label{app:#1}
} }

\title{Minimizing the Number of Tardy Jobs with Uniform Processing Times on Parallel Machines} 

\titlerunning{Minimizing the Number of Tardy Jobs with Uniform Processing Times} 

\author{Klaus Heeger}{Department of Industrial Engineering and Management, Ben-Gurion~University~of~the~Negev, 
	Beer-Sheva, 
	Israel}{heeger@post.bgu.ac.il}{https://orcid.org/0000-0001-8779-0890}{Supported by the ISF, grant No.~1070/20.}

\author{Hendrik~Molter}{Department of Computer Science, Ben-Gurion~University~of~the~Negev, 
	Beer-Sheva, 
	Israel}{molterh@post.bgu.ac.il}{https://orcid.org/0000-0002-4590-798X}{Supported by the European Union's Horizon Europe research and innovation programme under grant agreement 949707.}

\authorrunning{K. Heeger and H. Molter} 

\Copyright{Klaus Heeger and Hendrik Molter} 



\category{} 

\relatedversion{} 

\nolinenumbers 

\EventEditors{John Q. Open and Joan R. Access}
\EventNoEds{2}
\EventLongTitle{42nd Conference on Very Important Topics (CVIT 2016)}
\EventShortTitle{CVIT 2016}
\EventAcronym{CVIT}
\EventYear{2016}
\EventDate{December 24--27, 2016}
\EventLocation{Little Whinging, United Kingdom}
\EventLogo{}
\SeriesVolume{42}
\ArticleNo{23}

\begin{document}

\maketitle

\begin{abstract}
In this work, we study the computational (parameterized) complexity of \probweighted. Here, we are given $m$ identical parallel machines and $n$ jobs with equal processing time, each characterized by a release date, a due date, and a weight. The task is to find a feasible schedule, that is, an assignment of the jobs to starting times on machines, such that no job starts before its release date and no machine processes several jobs at the same time, that minimizes the weighted number of tardy jobs. A job is considered tardy if it finishes after its due date.

Our main contribution is showing that \probunweighted (the unweighted version of the problem) is NP-hard and W[2]-hard when parameterized by the number of machines. The former resolves an open problem in Note 2.1.19 by Kravchenko and Werner~[Journal of Scheduling, 2011] and Open Problem 2 by Sgall~[ESA, 2012], and the latter resolves Open Problem 7 by Mnich and van Bevern~[Computers \& Operations Research, 2018].
Furthermore, our result shows that the known XP-algorithm for \probweighted parameterized by the number of machines is optimal from a classification standpoint. 

On the algorithmic side, we provide alternative running time bounds for the above-mentioned known XP-algorithm. Our analysis shows that \probweighted is contained in XP when parameterized by the processing time, and that it is contained in FPT when parameterized by the combination of the number of machines and the processing time. Finally, we give an FPT-algorithm for \probweighted parameterized by the number of release dates or the number of due dates. 
With this work, we lay out the foundation for a systematic study of the parameterized complexity of~\probweighted.
\end{abstract}


\section{Introduction}
\label{sec:intro}

Machine scheduling is one of the most fundamental application areas of combinatorial optimization~\cite{pinedo2012scheduling}. In a typical scheduling problem, the task is to assign jobs to machines with the goal of maximizing a certain optimization objective while complying with certain constraints.
Jobs are usually characterized by a \emph{processing time}, a \emph{release date}, a \emph{due date}, and a \emph{weight} (or a subset thereof). We consider the setting where we have access to several identical parallel machines that can each process one job (non-preemptively) at a time. One of the most fundamental optimization objectives is to minimize the weighted number of tardy jobs, where a job is considered \emph{tardy} if it is completed after its due date.

The arguably simplest scheduling problem aims to minimize the (unweighted) number of tardy jobs on a single machine, where all jobs are released at time zero. In the standard three-field notation for scheduling problems by Graham~\cite{Graham1969}, this problem is called $1\mid \mid \sum U_j$. It can be solved in polynomial time by a classic algorithm of Moore~\cite{Moore68}. However, this problem becomes NP-hard when weights are introduced, the number of machines is increased, or release dates are added.
\begin{itemize}
    \item The weighted version, $1\mid \mid \sum w_j U_j$, is one of the first scheduling problems
shown to be (weakly) NP-hard. It remains hard even if all jobs have the same due date, as in this case, it is equivalent to the well-known \textsc{Knapsack} problem, which was already included in Karp’s famous list of 21 NP-hard problems~\cite{Kar72}. The problem can be solved in pseudopolynomial time with a classic algorithm by Lawler and Moore~\cite{LawlerMoore}.
\item Adding a second machine leads to the problem $2\mid \mid \sum_j U_j$, which is (weakly) NP-hard as it is a generalization of the well-known \textsc{Partition} problem, which was also already included in Karp’s list of 21 NP-hard problems~\cite{Kar72}. If the number of machines is unbounded, the problem is called $P\mid \mid \sum_j U_j$ and it is strongly NP-hard even if all jobs have the same due date, as it generalizes the well-known \textsc{Bin Packing} problem~\cite{GJ79}.
\item Introducing release times leads to the problem $1\mid r_j\mid \sum_j U_j$, which is weakly NP-hard~\cite{LenstraRinnooy-Kan77}, even if there are only two different release dates and two different due dates. The reduction of Lenstra et al.~\cite{LenstraRinnooy-Kan77} can be extended in a straightforward way to show that $1\mid r_j \mid \sum U_j $ is strongly NP-hard.
\end{itemize}

\subparagraph{Problem Setting and Motivation.} We consider the case where release dates and weights are present and we have multiple identical parallel machines. However, we add the restriction that all processing times are the same. This problem is called \probweighted, we give a formal definition in \cref{sec:prelims}. This problem arises naturally in manufacturing systems, where 
\begin{itemize}
    \item exact specifications by the customers for the product have negligible influence on the production time, but
    \item the specifications only become available at certain times and customers request the product to meet certain due dates.
\end{itemize}

As an illustrative example, consider the problem of scheduling burn-in operations in
integrated circuit manufacturing. The specification of the layout of the circuit may only become available at a certain time, as it takes time to optimize it. At the same time, the specific layout has little to no influence on the processing time of the burn-in operation~\cite{liu2009bicriterion}. Furthermore, the customer may wish to have the circuit delivered at a certain due date.

To the best of our knowledge, the only known algorithmic result for \probweighted is a polynomial-time algorithm by Baptiste et al.~\cite{baptiste2000scheduling,baptiste2004ten} for the case where the number of machines is a constant. 
However, two special cases are known to be polynomial-time solvable. It is folklore that the case without release dates, $P\mid p_j=p\mid \sum_j w_j U_j$, and the case where the processing times equal one, $P\mid r_j,p_j=1\mid \sum_j w_j U_j$, can both be reduced to the \textsc{Linear Assignment} problem in a straightforward manner. The \textsc{Linear Assignment} problem is known to be solvable in polynomial time~\cite{kuhn1955hungarian}.
Furthermore, it is known that, given an instance of \probweighted, we can check in polynomial time whether all jobs can be scheduled such that \emph{no} job is tardy~\cite{brucker2008scheduling,simons1983multiprocessor,simons1989fast}.

\subparagraph{Our Contribution.} The complexity status of \probweighted and its unweighted version \probunweighted was a longstanding open problem. Kravchenko and Werner~\cite{kravchenko2011parallel} pointed out that this question remains unanswered in Note 2.1.19 and Sgall~\cite{Sgall12} listed this issue as Open Problem~2. Our main contribution is to resolve the complexity status of \probunweighted (and hence also \probweighted) by showing the following.
\begin{itemize}
    \item \probunweighted is NP-hard.
\end{itemize}

 Having established NP-hardness, we focus on studying the parameterized complexity of \probweighted. As mentioned before, Baptiste et al.~\cite{baptiste2000scheduling,baptiste2004ten} showed that the problem is in XP when parameterized by the number of machines. Mnich and van Bevern~\cite{mnich2018parameterized} asked in Open Problem 7 whether this result can be improved to an FPT algorithm. We answer this question negatively by showing the following.
\begin{itemize}
    \item \probunweighted is W[2]-hard when parameterized by the number of machines.
\end{itemize}

On the positive side, we give several new parameterized tractability results. By providing an alternative running time analysis of the algorithm for \probweighted Baptiste et al.~\cite{baptiste2000scheduling,baptiste2004ten}, we show the following.
\begin{itemize}
    \item \probweighted is in XP when parameterized by the processing time.
    \item \probweighted is in FPT when parameterized by the combination of the number of machines and the processing time.
\end{itemize}

Finally, we give a new algorithm based on a mixed integer linear program (MILP) formulation for \probweighted. We show the following.
\begin{itemize}
    \item \probweighted is in FPT when parameterized by the number of different release dates.
    \item \probweighted is in FPT when parameterized by the number of different due dates.
\end{itemize}

We conclude by pointing to new future research directions. Most prominently, we leave open whether \probweighted is in FPT or W[1]-hard when parameterized by the processing time.

\subparagraph{Related Work.} 
We give an overview of the literature investigating the parameterized complexity of minimizing the weighted number of tardy jobs in various related settings.

The problem of minimizing the weighted number of tardy jobs on a single machine, $1\mid \mid \sum_j w_j U_j$ has been extensively studied in the literature under various aspects and constraints. Hermelin et al.~\cite{hermelin2023minimizing} showed that the classical pseudopolynomial time algorithm by Lawler and Moore~\cite{LawlerMoore} can be improved in several special cases.
Hermelin et al.~\cite{HermelinKPS21} give an overview of the parameterized complexity of $1\mid \mid \sum_j w_j U_j$ with respect to the parameters ``number of processing times'', ``number of due dates'', and ``number of weights'' (and their combinations).
In particular, $1\mid \mid\sum w_j U_j$ is in XP when parameterized by the number of different processing times~\cite{HermelinKPS21}. This presumably cannot be improved to an FPT result as recently, it was shown that $1\mid \mid\sum w_j U_j$ parameterized by the number of different processing times is W[1]-hard~\cite{HeegerH24}. Faster algorithms are known for the case where the job weights equal the processing times~\cite{bringmann2022faster} and the problem has also been studied under fairness aspects~\cite{heeger2023equitable}.

Minimizing the weighted number of tardy jobs on parallel machines has mostly been studied in the context of interval scheduling ($P\mid d_j-r_j=p_j \mid \sum_j w_j U_j$) and generalizations thereof~\cite{arkin1987scheduling,hermelin2024parameterized,krumke2011interval,sung2005maximizing}.

The setting where processing times are assumed to be uniform has been studied under various optimization criteria (different from minimizing the weighted number of tardy jobs) and constraints. For an overview, we refer to Kravchenko and Werner~\cite{kravchenko2011parallel} and Baptiste et al.~\cite{baptiste2004ten}. The even more special case of unit processing times has also been extensively studied. Reviewing the related literature in this setting is out of scope for this work.

\section{Preliminaries}\label{sec:prelims}
\subparagraph{Scheduling.} Using the standard three-field notation for scheduling problems by Graham~\cite{Graham1969}, the problem considered in this work is called \probweighted.
In this problem, we have $n$ jobs and $m$ machines.
Each machine can process one job at a time. Generally, we use the variable $j$ to denote jobs and the variable $i$ to denote machines.
Each job $j$ has a \emph{processing time} $p_j=p$, a \emph{release date} $r_j$, a \emph{due date} $d_j$, and a \emph{weight} $w_j$, where we $p$, $r_j$, $d_j$, and $w_j$ are nonnegative integers. We use $r_\#$, $d_\#$, and $w_\#$ to denote the number of different release dates, due dates, and weights, respectively. 

A schedule maps each job~$j$ to a combination of a machine~$i$ and a starting time~$t$, indicating that $j$ shall be processed on machine $i$ starting at time~$t$.
More formally, a \emph{schedule} is a function $\sigma : \{1,\ldots,n\}\rightarrow 
\{1,\ldots,m\}\times \mathbb{N}$. If for job $j$ we have $\sigma(j)=(i,t)$, then job $j$ is scheduled to be processed on machine $i$ starting at time $t$ until time $t+p$. 
A schedule $\sigma$ is \emph{feasible} if there is no job $j$ with $\sigma(j)=(i,t)$ and $t<r_j$ and if there is no pair of jobs~$j,j'$ with $\sigma(j)=(i,t)$ and $\sigma(j')=(i,t')$ such that $|t-t'|<p$. 
We say that a job $j$ is \emph{early} in a feasible schedule $\sigma$ if $\sigma(j)=(i,t)$ and $t+p\le d_j$, otherwise we say that job $j$ is \emph{tardy}. We say that machine $i$ is \emph{idle} at time $t$ in a feasible schedule $\sigma$ if there is no job $j$ with $\sigma(j)=(i,t')$ and $t'\le t\le t'+p$.
The goal is to find a feasible schedule that maximizes the weighted number of early jobs $W=\sum_{j\mid \sigma(j)=(i,t) \wedge t+p\le d_j}w_j$. We call a feasible schedule that maximizes the weighted number of early jobs \emph{optimal}. Formally, the problem is defined as follows.

\optprob{\probweighted}{A number $n$ of jobs, a number $m$ of machines, a processing time $p$, a set of release dates $\{r_1,r_2,\ldots,r_n\}$, a set of due dates $\{d_1,d_2,\ldots,d_n\}$, and a set of weights~$\{w_1,w_2,\ldots,w_n\}$.}{Find a feasible schedule $\sigma$ that maximizes $W=\sum_{j\mid \sigma(j)=(i,t) \wedge t+p\le d_j}w_j$.}

We use \probunweighted to denote the unweighted (or, equivalently, uniformly weighted) version of \probweighted, that is, the case where $w_j = w_{j'}$ for every two jobs~$j $ and $j'$, or equivalently, $w_\#=1$.

Given an instance $I$ of \probweighted, we can make the following observation. 
\begin{observation}\label{obs:releaseduedate}
    Let $I$ be an instance of \probweighted and let $d_{\max}$ be the largest due date of any job in $I$. Let $I'$ be the instance obtained from $I$ by setting $r'_j=d_{\max}-d_j$ and $d'_j=d_{\max}-r_j$. Then $I$ admits a feasible schedule where the weighted number of early jobs is~$W$ if and only if $I'$ admits a feasible schedule where the weighted number of early jobs is $W$.
\end{observation}
To see that \cref{obs:releaseduedate} is true note that a feasible schedule $\sigma$ for $I$ can be transformed into a feasible schedule $\sigma'$ for $I'$ (with the same weighted number of early jobs) by setting $\sigma'(j)=(i,d_{\max}-t-p)$, where $\sigma(j)=(i,t)$. Intuitively, this means that we can switch the roles of release dates and due dates.

We now show that we can restrict ourselves to schedules where jobs may start only at ``few'' different points in time, which will be useful in our proofs.
In order to do so, we define a set~$\mathcal{T}$ of \emph{relevant starting time points}.
\[
\mathcal{T}=\{t\mid \exists \ r_j \text{ and } \exists \ 0\le \ell\le n \text{ s.t.\ } t=r_j+ p\cdot \ell\}
\]
It is known that there always exists an optimal schedule where the starting times of all jobs are in $\mathcal{T}$.
\begin{lemma}[\cite{baptiste2000scheduling,baptiste2004ten}]\label{lem:relevant}
Let $I$ be an instance of \probweighted. Then there exists a feasible schedule $\sigma$ that maximizes the weighted number of early jobs such that for each job $j$ we have $\sigma(j)=(i,t)$ for some $t\in\mathcal{T}$.
\end{lemma}

\subparagraph{Parameterized Complexity.} We use the following standard concepts from parameterized complexity theory~\cite{Cyg+15,DF13,FG06}.
A \emph{parameterized problem}~$L\subseteq \Sigma^*\times \mathbb{N}$ is a subset of all instances~$(x,k)$ from~$\Sigma^*\times \mathbb N$, where~$k$ denotes the \emph{parameter}.
A parameterized problem~$L$ is 
in the class FPT (or \emph{fixed-parameter tractable}) if there is an algorithm that decides every instance~$(x,k)$ for~$L$ in~$f(k)\cdot |x|^{O(1)}$ time for some computable function~$f$ that depends only on the parameter. 
A parameterized problem~$L$ is 
in the class XP if there is an algorithm that decides every instance~$(x,k)$ for~$L$ in~$|x|^{f(k)}$ time for some computable function~$f$ that depends only on the parameter.
If a parameterized problem $L$ is W[1]-hard or W[2]-hard, then it is presumably not contained in FPT~\cite{Cyg+15,DF13,FG06}.

\section{Hardness of \boldmath\probunweighted}

In this section, we present our main result, namely that the unweighted version of our scheduling problem, \probunweighted, is NP-hard and W[2]-hard when parameterized by the number $m$ of machines. The former resolves an open problem in Note 2.1.19 by Kravchenko and Werner~\cite{kravchenko2011parallel} and Open Problem 2 by Sgall~\cite{Sgall12}, and the latter resolves Open Problem 7 by Mnich and van Bevern~\cite{mnich2018parameterized}.

\begin{theorem}\label{thm:unweighted-machines-w2}
    \probunweighted is NP-hard and W[2]-hard parameterized by the number~$m$ of machines.
\end{theorem}
    In order to show \Cref{thm:unweighted-machines-w2}, we present a parameterized reduction from \textsc{Hitting Set} parameterized by solution size~$k$, which is known to be NP-hard~\cite{Kar72} (unparameterized) and W[2]-hard~\cite{DF99}.
    
    \decprob{\textsc{Hitting Set}}{
    A finite set $U = \{u_0, \ldots, u_{n -1}\}$, a set~$\mathcal{A} = \{A_0, \ldots, A_{m-1}\}$ of subsets of $U$, and an integer $k$.
    }{
    Is there a \emph{hitting set} of size~$k$, that is, a set $X\subseteq U$ with $|X| = k$ and $X \cap A_j \neq \emptyset$ for every $j \in \{0, \ldots, m-1\}$?
    }
    
    Let $I=(U = \{u_0, \ldots, u_{n-1}\}, \mathcal A = \{A_0, \ldots, A_{m-1}\}, k)$ be an instance of \textsc{Hitting Set}. In order to ease the presentation, we give jobs names rather than identifying them with natural numbers. Furthermore, for a job $J$, we use $r(J)$ to denote the release date of $J$, and we use $d(J)$ to denote the deadline of $J$.
    
    Our reduction will have $k$ machines.
    The main idea behind the reduction is as follows.
    Each machine acts as an ``selection gadget'', that is, we will interpret the jobs scheduled on each machine in an optimal schedule as the selection of a particular element of $U$ to be included in the hitting set. 
    As there are $k$ machines, this ensures that the (hitting) set consisting of the selected elements has size at most $k$.
    Intuitively, we want that selecting element~$u_i$ on a machine corresponds to all jobs on this machine starting at time $i$ modulo $p$ in an optimal schedule.
    For each set~$A_j \in \mathcal{A}$, there are two kinds of jobs.
    \begin{itemize}
        \item First, jobs~$J_{A_j, u_i}$ for each $u_i \in A_j$, where scheduling job~$J_{A_j, u_i}$ encodes that the element $u_i$ is selected to be part of the hitting set.
        \item Second, there are $k-1$ dummy jobs~$D_{A_j}$ which can be scheduled on the up to $k-1$ machines not corresponding to elements of~$A_j$.
    \end{itemize}
    We give the jobs~$J_{A_j, u_i}$ and $D_{A_j}$ release dates and due dates such that they are the only jobs which can be started in the interval~$[j\cdot p, (j+1) \cdot p -1]$ and are early. Intuitively, this makes sure that an optimal solution has to schedule one of these jobs on each machine.
    In particular, one job~$J_{A_j,u_i}$ is scheduled, implying that $A_j$ is hit by one of the selected elements.

    There is, however, one problem with the reduction as sketched above.
    We do not ensure that all early jobs scheduled on some machine start at the same time modulo $p$.
    Thus, it is possible to schedule e.g.\ first job~$J_{A_1, u_1}$ on machine 1 and then job~$J_{A_2, u_4}$ such that both those jobs are early. Machine 1 now does not encode the selection of a single element to the hitting set.
    However, note that it is only possible to increase the index of the ``selected'' element, and as there are only $k$ machines, the total increase is bounded by $k \cdot (n-1)$.
    Consequently, repeating the reduction sketched above $k \cdot (n-1) + 1$ times ensures that at least one of the repeated instances will select only one item per machine, and then this instance correctly encodes a hitting set of size $k$.

    We now describe the reduction in detail. Formally, given the instance $I$ of \textsc{Hitting Set}, we construct an instance $I'$ of \probunweighted as follows.
    We set the processing time to~$p =2n$.
    We construct the following jobs for each $A_j \in \mathcal A $ and $\ell \in \{0, \ldots, k \cdot (n-1)\}$:
    \begin{itemize}
        \item one job~$J_{A_j, u_i}^\ell$ for each $u_i \in A_j$ with release date $r(J^\ell_{A_j, u_i}) = (\ell \cdot m  + j) \cdot p + i$ and due date $d(J^\ell_{A_j, u_i}) = r(J^\ell_{A_j, u_i}) + p$, and 
        \item $k-1 $ dummy jobs $D_{A_j}^\ell$ with release date $r(D_{A_j}^\ell) = (\ell \cdot m + j) \cdot p$ and due date $d(D_{A_j}^\ell) = r(D_{A_j}^\ell) + p + n$.
    \end{itemize}
    
    Finally, we set the number of machines to~$k$.
    This finishes the construction.
    For an overview of the due dates and release dates of the jobs, see \Cref{tab:jobs}. We can easily observe the following.
    \begin{table}[t]
        \centering
        $
        \begin{array}{c|c | c|c}
             \text{job} & \text{release date} & \text{due date} & \text{multiplicity}\\
             \hline\rule{0pt}{2.6ex}
             J^\ell_{A_j, u_i} & (\ell \cdot m + j) \cdot p + i & (\ell \cdot m  + j+1) \cdot p + i & 1\\
             \rule{0pt}{3ex}
             D^\ell_{A_j} & (\ell \cdot m + j) \cdot p & (\ell \cdot m + j + 1) \cdot p + n & k-1
        \end{array}
        $
        \smallskip
        \caption{Overview of the release dates and due dates of the jobs created for each $A_j \in \mathcal A $ and $\ell \in \{0, \ldots, k \cdot (n-1)\}$.}
        \label{tab:jobs}
    \end{table}

\begin{observation}\label{obs:reduction}
    Given an instance $I$ of \textsc{Hitting Set}, the above described instance $I'$ of  \probunweighted can be computed in polynomial time and has $k$ machines.
\end{observation}

    We continue by showing the correctness of the reduction. More specifically, we show that the \textsc{Hitting Set} $I$ instance is a yes-instance if and only if the constructed instance $I'$ of \probunweighted admits a feasible schedule with $(k\cdot (n-1) + 1) \cdot m \cdot k$ early jobs. 
We split the proof into the forward and backward direction.
    We start with the forward direction.
    
    \begin{lemma}\label{lem:hardness-forward}
        If the \textsc{Hitting Set} instance $I$ admits a hitting set of size $k$, then the \probunweighted instance $I'$ admits a feasible schedule with $(k\cdot (n-1) + 1) \cdot m\cdot k$ early jobs.
    \end{lemma}
    \begin{proof}
        Let $X = \{u_{i_1}, \ldots, u_{i_k}\}$ be a hitting set of size $k$ for $I$.
        We construct a schedule for~$I'$ as follows.
        On machine~$q$, for each $\ell \in \{0, \ldots, k \cdot (n-1)\}$ and $j\in \{0, \ldots, m-1\}$, we schedule one job to start at time~$(\ell \cdot m + j)\cdot p + i_q$.
        This job is $J^\ell_{A_j, u_{i_q}}$ if $u_{i_q} \in A_j$, and $D^\ell_{A_j}$ otherwise.
        Because $X$ is a hitting set, we have that $u_{i_q} \in A_j$ for some $q\in\{1,\ldots,k\}$ and hence we schedule each dummy job~$D^\ell_{A_j}$ at most $k-1$ times, that is, at most its multiplicity times.

        We can observe that all jobs scheduled so far are early. For each job $J^\ell_{A_j, u_{i_q}}$ that is scheduled on machine $q$, we have set its starting time to $(\ell \cdot m + j)\cdot p + i_q$ which equals this job's release date (cf.\ \cref{tab:jobs}). Furthermore, job $J^\ell_{A_j, u_{i_q}}$ finishes at $(\ell \cdot m + j+1)\cdot p + i_q$, its deadline. Each dummy job $D^\ell_{A_j}$ is early as well, since their release times are smaller or equal to the release time of $J^\ell_{A_j, u_{i_q}}$, and their due dates are larger or equal to the due date of $J^\ell_{A_j, u_{i_q}}$.
        Furthermore, we can observe that there is no overlap in the processing times between any two jobs scheduled on machine~$q$.

        It follows that we have feasibly scheduled $(k\cdot (n-1) + 1) \cdot m \cdot k$ such that they finish early. We schedule the remaining jobs in some arbitrary way such that the schedule remains feasible.
    \end{proof}

    Next, we continue with the backward direction.
    
    \begin{lemma}\label{lem:hardness-backward}
        If the \probunweighted instance $I'$ admits a feasible schedule with $(k\cdot (n-1) + 1) \cdot m\cdot k$ early jobs, then the \textsc{Hitting Set} instance $I$ admits a hitting set of size~$k$.
    \end{lemma}

    \begin{proof}
        Let $\sigma$ be a feasible schedule for $I'$ with $(k\cdot (n-1) + 1) \cdot m \cdot k$ early jobs.
        Since the largest due date is $(k \cdot (n-1) \cdot m + m) \cdot p + n$ and $p=2n$, it follows that on each machine, at most $k \cdot (n-1) \cdot m + m$ jobs can be scheduled in a feasible way such that they finish early.
        Consequently, on each machine, there are exactly $(k\cdot (n-1) + 1) \cdot m$ early jobs.
        
        More specifically, for each machine, each $\ell \in \{0, \ldots, k \cdot (n-1) \cdot m\} $ and $j \in \{0, \ldots , m-1\}$, there must be one job which starts in the interval $[(\ell \cdot m + j) \cdot p, \ell \cdot m\cdot p + (j-1) \cdot p + n]$. Otherwise, there would be less than $(k\cdot (n-1) + 1) \cdot m$ early jobs on that machine.
        For machine~$q$, let $x_\ell^q \in \{1,\ldots,n\}$ such that there is a job on machine~$q$ which starts at time $\ell \cdot m\cdot p + x_\ell^q$ (the existence of such a job is guaranteed by the previous observation).
        Then we must have $1\le x_0^q \le x_1^q \le \ldots x_{k \cdot (n-1) + 1}^q \le n$. This follows from the observation that if $x_\ell^q>x_{\ell+1}^q$, then the starting time of the job corresponding to $x_{\ell+1}^q$ would be earlier than the completion time of the job corresponding to $x_{\ell}^q$ and hence the schedule would be infeasible.
        Consequently, there are at most $n-1$ values of $\ell$ such that $x_\ell^q \neq x_{\ell  +1}^q$.
        As there are $k$ machines, this implies that there are at most $k \cdot (n-1) $ values such that $x_\ell^q  \neq x_{\ell + 1}^q $ for \emph{some} machine $q$.
        We can conclude that there exists at least one $\ell \in \{0, \ldots, k\cdot (n-1)\}$ so that $x_\ell^q = x_{\ell+1}^q$ for every machine~$q$.
        This implies that for each $j \in \{0, \ldots, m-1\}$ and each machine~$q$, there is one job starting at time $(\ell \cdot m + j)\cdot  p + x_\ell^q$ on machine $q$.
        We fix such an~$\ell $ for the rest of the proof and claim that $X=\{u_{x_\ell^1}, \ldots, u_{x_\ell^k}\}$ is a hitting set of size at most~$k$ for~$I$. 
        
        Clearly, $X$ has size at most $k$.
        Consider some set $A_j\in\mathcal{A}$ with $j \in \{0, \ldots, m -1\}$.
        The only jobs which can start in the interval $[(\ell \cdot m + j) \cdot p, (\ell \cdot m + j) \cdot p + n]$ are the dummy jobs~$D_{A_j}^\ell$ and the jobs~$J_{A_j, u_i}^\ell$ for $u_i \in A_j$.
        Because there are only $k-1$ dummy jobs~$D_{A_j}^\ell$ but $k$ machines, this implies that there exists at least one machine~$q$ such that job~$J_{A_j, u_i}^\ell$ is scheduled at machine $q$ for some $u_i \in A_j$.
        We know that on machine~$q$ one job starts in the interval $[(\ell \cdot m + j) \cdot p, (\ell \cdot m + j) \cdot p + n]$ and has starting at time $(\ell \cdot m + j) \cdot p + x_\ell^q$. By construction, this job must be $J_{A_j, u_i}^\ell$ with $i=x_\ell^q$ which implies that $u_i = u_{x_\ell^q}\in X$.
        We can conclude that $X$ is a hitting set for $I$.
    \end{proof}

Now we have all the pieces to prove \cref{thm:unweighted-machines-w2}.
\begin{proof}[Proof of \cref{thm:unweighted-machines-w2}]
\Cref{obs:reduction} shows that the described reduction can be computed in polynomial time and produces an instance of \probunweighted with $k$ machines. \Cref{lem:hardness-forward,lem:hardness-backward} show that the described reduction is correct. Since \textsc{Hitting Set} is known to be NP-hard~\cite{Kar72} and W[2]-hard when parameterized by $k$~\cite{DF99}, the result follows.
\end{proof}
\section{New Analysis of Known Algorithm for \boldmath\probweighted}

With \cref{thm:unweighted-machines-w2} we have established that \probunweighted is NP-hard. Hence, it is natural to resort to parameterized algorithms for efficiently finding exact solutions in restricted cases.
To the best of our knowledge, the only known parameterized algorithm for \probweighted is an XP-algorithm for the number $m$ of machines as a parameter by Baptiste et al.~\cite{baptiste2000scheduling,baptiste2004ten}.

\begin{theorem}[\cite{baptiste2000scheduling,baptiste2004ten}]\label{thm:xp-machines}
    \probweighted is in XP when parameterized by the number~$m$ of machines.
\end{theorem}

Since \cref{thm:unweighted-machines-w2} also shows W[2]-hardness for \probunweighted parameterized by the number $m$ of machines, the algorithm behind \cref{thm:xp-machines} presumably cannot be improved to an FPT-algorithm.

However, as it turns out, we can upper-bound the running time of the algorithm behind \cref{thm:xp-machines} in different ways to obtain additional tractability results. In the remainder of this section, we show that the algorithm developed by Baptiste et al.~\cite{baptiste2000scheduling,baptiste2004ten} additionally to \cref{thm:xp-machines} also implies the following.

\begin{theorem}\label{thm:xp-proctime}
    \probweighted is in XP when parameterized by the processing time~$p$ and \probweighted is in FPT when parameterized by the combination of the number~$m$ of machines and the processing time $p$.
\end{theorem}

In order to prove \cref{thm:xp-proctime}, we present the dynamic programming algorithm for \probweighted by Baptiste et al.~\cite{baptiste2000scheduling,baptiste2004ten}. For the correctness of this algorithm, we refer to their work. We give an alternative running time analysis that shows the claimed containments in XP and FPT.

To this end, we need to introduce some additional notation and terminology. 
Recall that~$\mathcal{T}$ denotes the set of relevant starting time points. The algorithm makes use of \cref{lem:relevant}, that is, we can assume the starting times of all jobs in an optimal schedule are from $\mathcal{T}$. A \emph{resource profile} is a vector $x=(x_1,x_2,\ldots,x_m)$ with $x_1\le x_2\le \ldots\le x_m$, $x_m-x_1\le p$, and $x_i\in \mathcal{T}$ for all $i\in\{1,\ldots,m\}$. Let $\mathcal{X}$ denote the set of all resource profiles. Now we define the following dynamic program. We assume that the jobs are sorted according to their due dates, that is, $d_1\le d_2\le\ldots\le d_n$. 

For two resource profiles $a,b\in\mathcal{X}$ and some $k\in\{1,\ldots,n\}$ we define $W(k,a,b)$ to be the maximum weighted number of early jobs of any feasible schedule for the jobs~$1, \ldots, k$ such that
\begin{compactitem}
    \item sorting the starting times of the first jobs on each machine from smallest to largest yields a vector $a'$ with $a\le a'$, and
    \item sorting the completion times of the last jobs on each machine from smallest to largest yields a vector $b'$ with $b'\le b$,
\end{compactitem}
where for two vectors $a,b$ of length $m$ we say that $a\le b$ if and only if for all $i\in\{1,\ldots,m\}$ we have that $a_i\le b_i$.

From this definition, it follows that $W(n,(0,\ldots,0),(t_{\max},\ldots,t_{\max}))$, where $t_{\max}$ is the largest element in $\mathcal{T}$, is the maximum weighted number of early jobs of any feasible schedule.
Baptiste et al.~\cite{baptiste2000scheduling,baptiste2004ten} proved the following.
\begin{lemma}[\cite{baptiste2000scheduling,baptiste2004ten}]\label{lem:dpcorr}
    For all $k\in\{1,\ldots,n\}$ and all resource profiles $a,b\in\mathcal{X}$ with $a\le b$ it holds that $W(k,a,b)$ equals $W(k-1,a,b)$ if $r_k\notin[a_m-p,b_1)$ and otherwise
    \[
    \max\left(W(k-1,a,b),\max_{x\in\mathcal{X},r_k\le x_1, x_1+p\le d_k, a\le x,\atop x'=(x_2,x_3,\ldots,x_m,x_1+p)\le b} \Bigl(W(k-1,a,x)+W(k-1,x',b)+w_k\Bigr)\right),
    \]
    where we define $W(0,a,b)=0$ for all $a,b\in\mathcal{X}$ with $a\le b$.
\end{lemma}
A straightforward running time analysis yields the following. We have that $|\mathcal{T}|\in O(n^2)$ and hence $|\mathcal{X}|\in O(n^{2m})$. It follows that the size of the dynamic programming table $W$ is in $O(n^{4m+1})$ and the time to compute one entry is in $O(n^{2m})$. This together with \cref{lem:dpcorr} yields \cref{thm:xp-machines}. In the remainder of the section, we give an alternative running time analysis to prove \cref{thm:xp-proctime}.
\begin{proof}[Proof of \cref{thm:xp-proctime}]
To prove \cref{thm:xp-proctime} we give a different bound for the size of $\mathcal{X}$. Recall that for all resource profiles $x\in\mathcal{X}$ we have that $x_m-x_1\le p$. It follows that there are $|\mathcal{T}|$ possibilities for the value of $x_1$, and then for $2\le i\le m$ we have that $x_1\le x_i\le x_1+p$. Hence, we get that $|\mathcal{X}|\in O(n^2\cdot p^{m-1})$. This together with \cref{lem:dpcorr} immediately gives us that \probweighted is in FPT when parameterized by $m+p$.

Furthermore, we have $x_1\le x_2\le \ldots\le x_m$. Hence, the resource profile $x$ can be characterized by counting how many times a value $t\in \mathcal{T}$ appears in $x$. Again, we can exploit that $x_m-x_1\le p$. There are $|\mathcal{T}|$ possibilities for the value of $x_1$ and given $x_1$, each other entry $x_i$ of $x$ with $2\le i\le m$ can be characterized by the amount $0 \le y_i=x_i-x_1\le p$ by which it is larger than $x_1$. Clearly, there are $p + 1$ different possible values for the amount~$y_i$. It follows that given $x_1$, we can characterize $x$ by counting how often a value $0\le t\le p$ appears as an amount. Hence, we get that $|\mathcal{X}|\in O(n^2\cdot m^{p+1})$.  This together with \cref{lem:dpcorr} immediately gives us that \probweighted is in XP when parameterized by $p$.
\end{proof}

Lastly, we remark that with similar alternative running time analyses for the dynamic programming algorithm by Baptiste et al.~\cite{baptiste2000scheduling,baptiste2004ten}, one can show that \probweighted is in XP when parameterized by the number of release dates or due dates, and that \probweighted is in FPT when parameterized by the combination of the number of machines and the number of release dates or due dates. However, as we show in the next section, we can obtain fixed-parameter tractability by parameterizing only by the number of release dates or parameterizing only by the number of due dates.

\section{FPT-Algorithm for \boldmath\probweighted}\label{sec:MILP}
\appendixsection{sec:MILP}

In this section, we present a new FPT algorithm for \probweighted parameterized by the number of release dates or due dates.
Formally, we show the following.

\begin{theorem}\label{thm:FPT}
    \probweighted\ is in FPT when parameterized by the number $r_\#$ of release dates or the number $d_\#$ of deadlines.
\end{theorem}

We prove \cref{thm:FPT} for the case where we parameterize by the number $r_\#$ of release dates. By \cref{obs:releaseduedate}, the case for the number $d_\#$ of deadlines is symmetric.
We present a reduction from \probweighted\ to \textsc{Mixed Integer Linear Program} (\textsc{MILP}).

\optprob{\textsc{Mixed Integer Linear Program} (\textsc{MILP})}{A vector $x$ of $n$ variables, a subset~$S$ of the variables which are considered integer variables, a constraint matrix $A\in\mathbb{R}^{m\times n}$, and two vectors $b\in\mathbb{R}^m$, $c\in \mathbb{R}^n$.}{Compute an assignment to the variables (if one exists) such that all integer variables in $S$ are set to integer values, $Ax\le b$, $x\ge 0$, and $c^\intercal x$ is maximized.}

We give a reduction that produces an \textsc{MILP} instance with a small number of integer values. More precisely, the number of integer values will be upper-bounded by a function of the number of release dates of the \probweighted instance. This allows us to upper-bound the running time necessary to solve the \textsc{MILP} instance using the following well-known result.
\begin{theorem}[\cite{dadush2011enumerative,Lenstra1983Integer}]\label{thm:MILP}
    \textsc{MILP} is in FPT when parameterized by the number of integer variables.
\end{theorem}

Furthermore, we construct the \textsc{MILP} instance in a way that ensures that there always exist optimal solutions where \emph{all} variables are set to integer values. Informally, we ensure that the constraint matrix for the rational variables is totally unimodular\footnote{A matrix is \emph{totally unimodular} if each of its square submatrices has determinant $0$, $1$, or $-1$~\cite{dantzig1956linear}.}. This allows us to use the following result.

\begin{lemma}[\cite{CMZ24}]\label{lem:MILP}
    Let $A_{\text{frac}} \in \mathbb{R}^{m \times n_2}$ be totally unimodular.
    Then for any $A_{\text{int}} \in \mathbb{R}^{m \times n_1}$, $b \in \mathbb{R}^m$, and $c \in \mathbb{R}^{n_1 + n_2}$, the MILP
    \[
    \max c^\intercal x \text{ subject to } (A_{\text{int}} \ A_{\text{frac}})x\le b, x\ge 0,
    \]
    where $x = (x_{\text{int}} \ x_{\text{frac}})^\intercal$ with the first $n_1$ variables (i.e., $x_{\text{int}}$) being the integer variables, has an optimal solution where all variables are integer.
\end{lemma}

Before we describe how to construct an MILP instance for a given instance of \probweighted, we make an important observation on optimal schedules.
    Intuitively, we show that we can assume that each job is scheduled as early as possible and idle times only happen directly before release dates.
\begin{lemma}\label{lem:releasedates}
    Let $\sigma$ be a feasible schedule for an instance of \probweighted such that the weighted number of early jobs is $W$. Then there exists a feasible schedule $\sigma'$ such that
    \begin{itemize}
        \item the weighted number of early jobs is $W$,
        \item for each job $j$ with $\sigma'(j)=(i,t)$ for some $t \neq r_j$, machine $i$ is not idle at time $t-1$, and
        \item all starting times of $\sigma'$ are in the set~$\mathcal{T}$ of relevant starting time points.
    \end{itemize}
\end{lemma}
\begin{proof}
    Assume that there is a job $j$ such that $\sigma(j)=(i,t)$ for some $t < r_j$ and machine $i$ is idle at time $t-1$. Assume that job $j$ is the earliest such job, that is, the job with minimum~$t$. Since machine $i$ is idle at time $t-1$ and $r_j<t$, we can create a new schedule $\sigma'$ that is the same as $\sigma$ except that $\sigma'(j)=(i,t-1)$. Clearly, we have that $\sigma'$ is feasible and has the same weighted number of early jobs as $\sigma$. By repeating this process, we obtain a feasible schedule $\sigma''$ with the same set of early jobs and such that for each job $j$ with $\sigma''(j)=(i,t)$ and machine $i$ is idle at time $t-1$, it holds that $t=r_j$.
    Furthermore, we have that each starting point in $\sigma''$ is a release date $r$ or a time $t$ with $t=r+\ell\cdot p$ for some release date $r$ and some integer $\ell$. Hence, all starting times of $\sigma''$ are in the set~$\mathcal{T}$ of relevant starting time points.
\end{proof}
We call a feasible schedule $\sigma$ \emph{release date aligned} if the second condition of \Cref{lem:releasedates} holds, i.e., for each job $j$ with $\sigma(j)=(i,t)$ for some $t > r_j$, the machine $i$ is idle at time $t-1$. Note that \cref{lem:releasedates} is stronger than \cref{lem:relevant} and implies that there always exists an optimal feasible schedule that is release date aligned.

    Given a feasible schedule $\sigma$ that is release date aligned, we say that a release date $r_j$ is \emph{active} on machine $i$ if job $j$ is scheduled to starts at this release date, that is, $\sigma(j)=(i,r_j)$, and machine $i$ is idle at time $r_j-1$. Let $T$ be a subset of all release dates, then we say that machine $i$ has type $T$ in $\sigma$ if all release dates in $T$ are active on machine $i$.

Recall that~$\mathcal{T}=\{t\mid \exists \ r_j \text{ and } \exists \ 0\le \ell\le n \text{ s.t.\ } t=r_j+p\cdot \ell\}$ denotes the set of relevant starting time points.
 We say that a starting time $t\in\mathcal{T}$ is available on a machine with type $T$ if $t=r+ \ell\cdot p$ for some $r\in T$ and $t+p\le r'$, where $r'$ is the smallest release date in $T$ that is larger than $r$.
    
Given an instance $I$ of \probweighted, we create an instance $I'$ of MILP as follows.
     For each type $T$ we create an integer variable $x_T$ that quantifies how many machines have type $T$ and create the constraint
    \begin{equation}\label{const:1}
        \sum_T x_T\le m.
    \end{equation}
    
    For each starting time $t\in\mathcal{T}$ we create a fractional variable $x_t$ that quantifies on how many machines the starting time $t$ is available and create the following set of constraints.
    \begin{equation}\label{const:2}
        \forall \ t \in \mathcal{T}: \ x_t=\sum_T t_T\cdot x_T,
    \end{equation}
    where $t_T=1$ if starting time $t$ is available on a machine with type $T$ and $t_T=0$ otherwise.

    For each combination of a job $j$ and a starting time $t$, we create a fractional variable~$x_{j,t}$ if job~$j$ can be scheduled to start at time $t$ without violating $j$'s release date or due date. This variable indicates whether job $j$ is scheduled to start at starting time $t$ and is early. We create the following constraints.
    \begin{align}
        & \forall \ t\in\mathcal{T} : \ \sum_j x_{j,t} \le x_t. \label{const:3} \\
       & \forall \ j\in\{1,\ldots,n\} : \ \sum_t x_{j,t} \le 1. \label{const:4}
    \end{align}

    Finally, we use the following function as the maximization objective.
    \begin{equation}\label{obj}
        \sum_{j,t} w_j \cdot x_{j,t}
    \end{equation}
This finishes the construction of the MILP instance $I'$. We can observe the following.
\begin{observation}\label{obs:MILPsize}
    Given an instance $I$ of \probweighted, the above described MILP instance $I'$ can be computed in $O(2^{r_\#}\cdot (m+n)^2)$ time and has $O(2^{r_\#})$ integer variables.
\end{observation}

In the following, we prove the correctness of the reduction to MILP. We start with showing that if the \probweighted instance admits a feasible schedule where the weighted number of early jobs is $W$, then the constructed MILP instance admits a feasible solution that has objective value $W$.

\begin{lemma}
\label{lem:MILPcorr1}
    If the \probweighted instance $I$ admits a feasible schedule where the weighted number of early jobs is $W$, then the MILP instance $I'$ admits a feasible solution that has objective value $W$.
\end{lemma}
\begin{proof}
    Assume that we are given a feasible schedule $\sigma$ for $I$ such that the weighted number of early jobs is $W$. By \cref{lem:releasedates} we can assume that $\sigma$ is release date aligned.

    We construct a solution for $I'$ as follows. Consider job $j$ and let $\sigma(j)=(i,t)$ such that $t+p\le d_j$, that is, job $j$ is early. We know that $t\in\mathcal{T}$. We set $x_{j,t}=1$ and for all $t'\in\mathcal{T}$ with $t'\neq t$, we set $x_{j,t'}=0$. Note that this guarantees that Constraints~(\ref{const:4}) are fulfilled. Furthermore, assuming we can set the remaining variables to values such that the remaining constraints are fulfilled, we have that the objective value of the solution to $I'$ is $W$.

    In the remainder, we show how to set the remaining variables such that all constraints are fulfilled. Initially, we set all variables $x_T$ to zero. Next, we determine the type of each machine.  Consider machine $i$. Then $T :=\{r\mid \exists j \text{ s.t.\ } \sigma(j)=(i,r_j)\}$ is the type of machine $i$. Then we increase $x_T$ by one. We do this for every machine. Clearly, afterwards Constraint~(\ref{const:1}) are fulfilled.

    Next, we set $x_t:=\sum_T t_T\cdot x_T$, where $t_T=1$ if starting time $t$ is available on a machine with type $T$ and $t_T=0$ otherwise. Clearly, this fulfills Constraints~(\ref{const:2}). It remains to show that Constraints~(\ref{const:3}) are fulfilled.
    So consider some time $t\in \mathcal T$.
    For each job~$j$ starting at time~$t$ on some machine~$i$, we increased $x_T$ by one for some~$T $ with $t_T = 1$ when processing machine $i$.
    Thus, we have $\sum_j x_{j, t} \le \sum_T t_T \cdot x_T = x_t$.
\end{proof}

Before we continue with the other direction of the correctness, we prove that we can apply \cref{lem:MILP} to show that the MILP instance $I'$ admits an optimal solution where all variables are set to integer values.

\begin{lemma}\label{lem:fracTU}
The MILP instance $I'$ admits an optimal solution where all variables are set to integer values.
\end{lemma}
\begin{proof}
Notice that since the Constraints~(\ref{const:2}) are equality constraints, we have that in any feasible solution to $I'$, all variables $x_t$ are set to integer values. Hence, $I'$ is equivalent to the MILP~$I''$ arising from~$I'$ by declaring $x_t$ to be integer variables for every $t \in \mathcal{T}$ (in addition to the variables $x_T$), and it suffices to show that $I''$ has an integer solution.

We show that the constraint matrix for the fractional variables $x_{j,t}$ in~$I''$ (i.e., the variables~$x_{j,t}$) is totally unimodular. By \cref{lem:MILP} this implies that $I''$ and therefore also $I'$ admits an optimal solution where all variables are set to integer values.

     Note that the Constraints~(\ref{const:3}) partition the set of fractional variables $x_{j,t}$, that is, each fractional variable is part of exactly one of the Constraints~(\ref{const:3}). The same holds for the Constraints~(\ref{const:4}). Furthermore, the coefficients in the constraint matrix for each variable are either 1 (if they are part of a constraint) or 0. Hence, we have that the constraint matrix is a 0-1 matrix with exactly two 1's in every column. Moreover, in each column, one of the two 1's appears in a row corresponding to the Constraints~(\ref{const:3}) and the other 1 is in a row corresponding to the Constraints~(\ref{const:4}). This is a sufficient condition for the constraint matrix to be totally unimodular~\cite{dantzig1956linear}.
\end{proof}

Now we proceed with showing that if the constructed MILP instance admits an optimal solution that has objective value~$W$, then the original \probweighted instance admits a feasible schedule where the weighted number of early jobs is $W$.

\begin{lemma}\label{lem:MILPcorr2}
    If the MILP instance $I'$ admits an optimal solution that has objective value~$W$, then the \probweighted instance $I$ admits a feasible schedule $\sigma$ where the weighted number of early jobs is $W$. The schedule $\sigma$ can be computed from the optimal solution to $I'$ in polynomial time (in the size of $I'$).
\end{lemma}
\begin{proof}
    Assume we are given an optimal solution for $I'$ that has objective value $W$. By \cref{lem:fracTU} we can assume that the optimal solution sets all variables to integer values. We construct a feasible schedule $\sigma$ as follows.

    First, we assign a type to every machine. Iterate through the types $T$ and, initially, set $i=1$. If $x_T>0$, then assign type $T$ to machines $i$ to $i+x_T-1$. Afterwards, increase $i$ by $x_T$. Since the solution to $I'$ fulfills Constraint~(\ref{const:1}), we know that this procedure does not assign types to more than $m$ machines.

    Now iterate through the relevant starting times $i\in \mathcal{T}$. Let $J_t=\{j\mid x_{j,t}=1\}$ and let $M_t=\{i \mid \text{starting time } t \text{ is available on machine }i\}$. By Constraints~(\ref{const:2}) and~(\ref{const:3}) we know that $|J_t|\le |M_t|=x_t$. Hence, we can create a feasible schedule by setting $\sigma(j)=(i,t)$ for every job $j\in J_t$, where $i\in M_t$ and for all $j,j'\in J_t$ with $j\neq j'$ we have $\sigma(j)=(i,t)$ and $\sigma(j')=(i',t)$ with $i\neq i'$. In other words, we schedule each job in $J_t$ to a distinct machine $i\in M_t$ with starting time $t$. The Constraints~(\ref{const:4}) ensure that we schedule each job at most once. For all jobs $j$ that are not contained in any set $J_t$ with $t\in \mathcal{T}$, we schedule $j$ to an arbitrary machine $i$ to a starting time that is later than $r_j$ and later than the completion time of the last job scheduled on $i$. Clearly, we can compute $\sigma$ in time polynomial in $|I'|$. By the definition of available starting times and the fact that we only create variable $x_{j,t}$ if $t\ge r_j$, this schedule is feasible. 
    
    It remains to show that the weighted number of early jobs is~$W$. To this end, note that we only create variable $x_{j,t}$ if $r_j\le t$ and $t+p\le d_j$. Hence, for each job $j$ with $x_{j,t}=1$ for some $t\in\mathcal{T}$, we know that this job is early in the constructed schedule $\sigma$. It follows that the weighted number of early jobs is $\sum_{j,t} w_j \cdot x_{j,t}$, which equals the maximization objective of~$I$ and hence equals $W$.
\end{proof}

Now we have all the pieces to prove \cref{thm:FPT}.
\begin{proof}[Proof of \cref{thm:FPT}]
Given an instance $I$ of \probweighted we create an MILP instance $I'$ as described above and use \cref{thm:MILP} to solve it. \cref{lem:MILPcorr1,lem:MILPcorr2} show that we can correctly compute an optimal schedule for $I$ from the solution to $I'$ in polynomial time (in the size of $I'$).
\cref{obs:MILPsize} together with \cref{thm:MILP} show that this algorithm has the claimed running time upper-bound.
\end{proof}

\section{Conclusion and Future Work}

In this work, we resolved open questions by Kravchenko and Werner~\cite{kravchenko2011parallel}, Sgall~\cite{Sgall12}, and Mnich and van Bevern~\cite{mnich2018parameterized} by showing that \probunweighted is NP-hard and W[2]-hard when parameterized by the number of machines.
The established hardness of the problem motivates investigating it from the viewpoint of exact parameterized or approximation algorithms.
In this work, we focussed on the former, leaving the latter for future research.
We provided a first step in systematically exploring the parameterized complexity of \probweighted. Our parameterized hardness result shows that the known XP-algorithm for the number of machines as a parameter is optimal from a classification standpoint. Furthermore, we showed that this known algorithm implies that the problem is also contained in XP when parameterized by the processing time, and that it is contained in FPT when parameterized by the combination of the number of machines and the processing time. Finally, we give an FPT-algorithm for \probweighted parameterized by the number of release dates (or due dates). We leave several questions open, the most interesting one is the following. 
\begin{itemize}
    \item Is \probweighted in FPT or W[1]-hard when parameterized by the processing time?
\end{itemize}

Other interesting parameters to consider might be the number of early jobs or the number of tardy jobs. It is easy to see that \probweighted is in XP when parameterized by either one of those parameters, by some simple guess-and-check algorithm (recall that we can check in polynomial time whether all jobs can be scheduled early~\cite{brucker2008scheduling,simons1983multiprocessor,simons1989fast}). Hence, it remains open whether the problem is in FPT or W[1]-hard with respect to those parameters.



\bibliography{bibliography}


\end{document}